\documentclass[11pt]{article}
\usepackage{amssymb}
\usepackage{colortbl}
\usepackage{amsfonts,amsmath, longtable}

\newcommand{\mR}{{\mathcal R}}

\newcommand{\mC}{\mathbb C}
\newcommand{\mZ}{\mathbb Z}

\newtheorem{lemma}{Lemma}

\newenvironment{proof}{\par\noindent{\bf Proof.}}{\hfill$\scriptstyle\blacksquare$}

\def\beq{\begin{equation}}
\def\eq{\end{equation}}

\newtheorem{prop}{Proposition}

\begin{document}

\setcounter{page}{1}

\begin{center}

\

\vspace{-0mm}

{\Large{\bf A solution of the associative Yang-Baxter equation}}

 \vspace{3mm}

{\Large{\bf related to the queer Lie superalgebra}}

 \vspace{15mm}

 {\Large {Maria Matushko}}
\qquad\quad\quad

  \vspace{5mm}

{\em Steklov Mathematical Institute of Russian
Academy of Sciences,\\ Gubkina str. 8, 119991, Moscow, Russia}

{\em Center for Advanced Studies, Skoltech, 143026, Moscow, Russia}


   \vspace{5mm}


\end{center}

\vspace{0mm}

\begin{abstract}
We propose a trigonometric solution of the associative Yang-Baxter equation related to the queer Lie superalgebra which in its turn satisfies the quantum Yang-Baxter equation.  
\end{abstract}

\newpage

\section{Introduction}\label{sec1}
\setcounter{equation}{0}
Let $A$ be an associative algebra and $R(\hbar,u)$ is a meromorphic function of two complex variables $(\hbar, u)$ 
taking values in $A\otimes A$. The {\em associative Yang-Baxter equation} (AYBE) is the equation in $A\otimes A\otimes A$:

\beq\label{AYBE1}
\begin{array}{c}
    R_{12}(x,u) R_{23}(y,v) - R_{13}(y,u+v) R_{12}(x-y,u) - R_{23}(y-x,v) R_{13}(x,u+v)=0\,.
\end{array}
\eq
The AYBE was introduced in \cite{Agu} and independently in \cite{Pol}. We remark also that the similar constant AYBE emerges in \cite{FK}. For $A={{\rm End}} (\mC^{N})$ the solutions of (\ref{AYBE1}) were classified in \cite{Pol2},\cite{Sch}, for instance the quantum elliptic Baxter-Belavin $R$-matrix satisfies (\ref{AYBE1}) in some proper normalization. In scalar case $A={{\rm End}} (\mC)$ the equation reduces to the functional equation which is known as the Fay identity. 
\beq\label{Fay}
    \phi(x,u) \phi(y,v) - \phi(y,u+v) \phi(x-y,u) - \phi(y-x,v) \phi(x,u+v)=0\,.
\eq
In the elliptic case the solution of (\ref{Fay}) is the following $$\phi(\hbar,u)=\frac{\theta'(0)\theta(u+\hbar)}{\theta(\hbar)\theta(u)},$$ where $\theta(z)=\theta_{11}(z|\tau)$ is the odd Riemann theta function; in trigonometric case it reduces to: 
 \beq\label{phi}
\phi(\hbar,u)=\pi \cot(\pi \hbar)+ \pi \cot(\pi u).
 \eq
 
 Let us write the AYBE equation in the following form:
 \beq\label{AYBE}
\begin{array}{c}
    R^{x}_{12}(u_1,u_2) R^{y}_{23}(u_2,u_3) - R^{y}_{13}(u_1,u_3) R^{x-y}_{12}(u_1,u_2) - R^{y-x}_{23}(u_2,u_3) R^{x}_{13}(u_1,u_3)=0\,.
\end{array}
\eq
In $  R^{h}_{ij}(u_i,u_j)$ the parameter  $\hbar$ plays the role of the Plank constant and $u_i,u_j$ are spectral parameters. To get (\ref{AYBE1}) from (\ref{AYBE}) we just put $R^{\hbar}_{ij}(u_i,u_j):=R_{ij}(\hbar,u_i-u_j)$ with $u=u_1-u_2$ and $v=u_2-u_3$. Thus (\ref{AYBE1}) gives the solutions $R^{\hbar}(u,v)$ of (\ref{AYBE})  depending on the difference $u-v$. 

The associative Yang-Baxter equation is closely related with the classical Yang-Baxter equation and the quantum Yang-Baxter equation. Following \cite{Pol},\cite{Sch} if some AYBE solution satisfies the condition 
 \beq\label{skewsym}
R_{12}^\hbar(u,v)=-R_{21}^{-\hbar}(v,u).
\eq
and has a Laurent expansion near $h=0$ of the form:
\beq\label{expans}
R_{12}^\hbar(u,v)=\frac{1}{\hbar}+r_{12}(u,v)+\hbar\, m_{12}(u,v)+O(\hbar^2),
\eq
then 
$r_{12}(u,v)$ is the antisymmetric ($r_{12}(u,v)=-r_{21}(v,u)$) solution of the {\em classical Yang-Baxter equation}: 
\beq\label{clYBE}
[r_{12}(u,v), r_{13}(u,w)]+[r_{12}(u,v), r_{23}(v,w)]+[r_{13}(u,w), r_{23}(v,w)]=0.
\eq
In \cite{Pol} (see also Lemma \ref{Lemma1}) it is shown that often the solutions of the AYBE are automatically the solutions of the {\em quantum Yang-Baxter equation} with the fixed $\hbar$:
 \beq\label{QYB}
\displaystyle{
    R^{\hbar}_{12}(u,v)  R^{\hbar}_{13}(u,w) R^{\hbar}_{23}(v,w) =
      R^{\hbar}_{23}(v,w) R^{\hbar}_{13}(u,w) R^{\hbar}_{12}(u,v)\,.
      }
\eq
The solution of (\ref{QYB}) is called a {\em $R$-matrix}.

Further we present the $R$-matrices satisfying the AYBE (\ref{AYBE}) and related to the graded algebras. Note that these solutions of (\ref{AYBE}) depend not only on the difference of spectral parameters.

Consider the general linear Lie superalgebra ${\rm gl}_{N|N}$ with standard generators $e_{ij}$. Let the indices $i,j$ run through $\pm 1,\dots,\pm N$.
Consider the $\mZ_2$-graded
vector space $\mC^{N|N}=\mC^{N}\oplus\mC^{N}$ with grading
\beq \label{grading}
p_i=\begin{cases}
     0 \text{     for    } 1\le i\le N, \\
     1 \text{     for   } -N\le i\le -1.
    \end{cases}
\eq
Let $e_{ij}\in{{\rm End}} (\mC^{N|N})$ be the standard matrix units, the algebra ${{\rm End}} (\mC^{N|N})$ is $\mZ_2$-graded so that ${\rm deg}\, e_{ij}=p_i+p_j$.

Let $R^\hbar(u,v)$ be the rational function of two complex variables $u,v$ valued in ${{\rm End}} (\mC^{N|N})^{\otimes 2}$:
\beq\label{ratR}
\begin{array}{c}
   \displaystyle{
R^\hbar_{12}(u,v)=\frac{1}{\hbar}+\sum_{i,j} e_{ij}\otimes e_{ji} \frac{(-1)^{p_j}}{u-v}+\sum_{i,j} e_{ij}\otimes e_{-j,-i} \frac{(-1)^{p_j}}{u+v}
 }
  \end{array}
\eq

The function (\ref{ratR}) was introduced in \cite{Nazold}. It satisfies the quantum Yang-Baxter equation (\ref{QYB}) in ${{{\rm End}} (\mC^{N|N})^{\otimes 3}}(u,v,w)$ and is called the quantum rational  $R$-matrix for the queer Lie superalgebra. In its turn the $R$-matrix (\ref{ratR}) satisfies the associative Yang-Baxter equation (\ref{AYBE}) we prove this in Section 2.2.
Note that the AYBE equation depends on the normalization while the solutions of the quantum Yang-Baxter equation (\ref{QYB}) are defined up to multiplication by a function depending on $\hbar$ and spectral parameters.

We present a trigonometric quantum $R$-matrix satisfying the AYBE and related to the queer Lie superalgebra. This solution is given by the following formula:
 \begin{equation}\label{superRqueer}
   \begin{array}{c}
   \displaystyle{
  {\rm R}^\hbar_{12}(u,v)
  =\pi\sum\limits_{i}\Big((-1)^{p_i}\cot(\pi (u-v))+\cot(\pi\hbar)\Big) e_{ii}\otimes e_{ii}+
   }
     \\ \ \\
   \displaystyle{
  \pi\sum\limits_{i}\Big((-1)^{p_i}\cot(\pi (u+v))+\cot(\pi\hbar)\Big) e_{ii}\otimes e_{-i,-i}+
   }
  \\ \ \\
     \displaystyle{
  +\pi\sum\limits_{|i|\neq |j|} e_{ii}\otimes e_{jj}\frac{ \exp\Big(\frac{\pi\imath\hbar}{N}\Big(\left(|i|-|j|\right)-N{\rm sign}\left(|i|-|j|\right)\Big)\Big)}{\sin(\pi\hbar)}+
  }
  \\ \ \\
   \displaystyle{
 +\pi\sum\limits_{i\neq j} (-1)^{p_j} e_{ij}\otimes e_{ji}\,\frac{ \exp\Big(\frac{\pi\imath (u-v)}{N}\Big((i-j)-N{\rm sign}(i-j)\Big)\Big)}{\sin(\pi (u-v))}\,+
  }
   \\ \ \\
   \displaystyle{
 +\pi\sum\limits_{i\neq j} (-1)^{p_j} e_{i,j}\otimes e_{-j,-i}\,\frac{ \exp\Big(\frac{\pi\imath (u+v)}{N}\Big((i-j)-N{\rm sign}(i-j)\Big)\Big)}{\sin(\pi (u+v))}\,.
  }
  \end{array}
  \end{equation}
It also satisfies the quantum Yang-Baxter equation. The proofs are given in Section \ref{trR}. A trigonometric solution of the quantum Yang-Baxter equation related to the queer Lie superalbegra without spectral parameters was found in \cite{Olsh}. In \cite{JNaz} there is a recipe how to insert the spectral parameters in this case. We present the explicit formula (\ref{uqsuperqueer}) of the $R$-matrix constructed in this way. However, this $R$-matrix does not satisfy the AYBE. We find the corresponding twist to connect both the solutions  (\ref{uqsuperqueer}) and (\ref{superRqueer})  of the quantum Yang-Baxter equation.  

A trigonometric solution of (\ref{AYBE}) for the graded case depending on difference $u-v$ is presented in \cite{MZ2}. Here we construct the trigonometric solution (\ref{superRqueer}), which depends not only on the difference $u-v$. This result extends the classification of the trigonometric solutions of AYBE given in \cite{Pol2},\cite{Sch}. Also, solutions of the AYBE are used for different constructions in integrable systems \cite{FK},\cite{Kiril},\cite{GrSZ},\cite{LOZ14},\cite{KrZ},\cite{SeZ},\cite{SeZ3},\cite{SeZ2}.


\section{Rational solution of the AYBE}
\setcounter{equation}{0}
\subsection{The queer Lie superalgebra}

Let $A$ and $B$ be any associative $\mZ_2$-graded algebras. Their tensor product is also a $\mZ_2$-graded algebra with the induced  $\mZ_2$-grading, where the operation for the homogeneous elements $a_1,a_2\in A, b_1,b_2\in B$ is defined by
\beq\label{conv1}
(a_1\otimes b_1)(a_2\otimes b_2)=(-1)^{(\deg a_2) (\deg b_1)}a_1a_2\otimes b_1b_2.
\eq
For any two $\mZ_2$-graded modules $V_1$ and $V_2$ over $A$ and $B$ respectively, the vector space  $V_1\otimes V_2$ is a $\mZ_2$-graded module over $A\otimes B$ such that for any homogeneous elements  $v_1\in V_1$ and $v_2\in V_2$
\beq\label{conv2}
\begin{array}{ccc}
(a\otimes b)(v_1\otimes v_2)=(-1)^{\deg b\deg v_1} a v_1\otimes b v_2\,,\\ \\
\deg(v_1\otimes v_2)=\deg(v_1)+\deg (v_2)\,.
\end{array}
\eq
Consider the $\mZ_2$-graded
vector space $\mC^{N|N}=\mC^{N}\oplus\mC^{N}$ with the standard basis $e_i, i=\pm 1,\dots,\pm N$ and grading (\ref{grading}), $\deg e_i=p_i$. Let $e_{ij}\in {{\rm End}} (\mC^{N|N})$ be the standard matrix unit, defined by $e_{ij}e_k=\delta_{jk}e_i$. The associative algebra ${{\rm End}} (\mC^{N|N})$ is $\mZ_2$-graded so that  $\deg e_{ij}=p_i+p_j$ and $\mC^{N|N}$ is a $\mZ_2$-graded module over ${{\rm End}} (\mC^{N|N})$. Using (\ref{conv1}) and (\ref{conv2}) the tensor product $\left(\mC^{N|N})\right)^{\otimes k}$ is a $\mZ_2$-graded module over ${{{\rm End}}} (\mC^{N|N})^{\otimes k}$.  
We will also regard $e_{ij}$ as generators of the complex Lie superalgebra ${\rm gl}_{N|N}$.

Define an involutive automorphism $\eta$ of ${{{\rm End}}} (\mC^{N|N})$ by mapping
\beq\label{autom}
\eta: e_{i,j}\to e_{-i,-j}. 
\eq
The {\em queer Lie superalgebra} ${\rm q}_N$ is the fixed point subalgebra in ${\rm gl}_{N|N}$ with respect to the involutive automorphism $\eta$.

Consider the element 
\beq\label{J}
J=\sum_i e_{i,-i}\, (-1)^{p_i},
\eq
of algebra ${{\rm End}} (\mC^{N|N})$, note  that $\deg (J)=1$. The following relations hold 
\beq\label{Jrel}
J^2=-{\rm Id} \qquad J_a J_b=-J_b J_a ,
\eq
where $J_a={\rm Id}^{\otimes (a-1)}\otimes J\otimes {\rm Id}^{\otimes (m-1-a)} $ for positive integers $a\le m$. 
Let ${{\rm Q}}(N)\subset {{\rm End}} (\mC^{N|N})$ be the centralizer of $J$, it is spanned by elements 
$$
e_{ab}+e_{-a,-b}, \qquad e_{a,-b}+e_{-a,b}.
$$
The subalgebra $Q(N)\subset {{\rm End}} (\mC^{N|N})$ can be viewed as the Lie superalgebra ${\rm q}_N$.


\subsection{R-matrix properties}
By $P_{12}$ we denote the superpermutation operator in  ${{\rm End}} (\mC^{N|N})^{\otimes 2}$:
\beq\label{permutation}
P_{12}=\sum_{i,j} e_{i,j}\otimes e_{ji}\, (-1)^{p_j}.
\eq
Let us rewrite the rational $R$-matrix (\ref{ratR}) in new notations (\ref{permutation}) and (\ref{J}), see also \cite{Naz} :
\beq\label{ratR2}
\begin{array}{c}
   \displaystyle{
R^\hbar_{12}(u,v)=\frac{1}{\hbar}+ \frac{P_{12}}{u-v}+({\rm Id} \otimes \eta)\frac{P_{12}}{u+v}=
 }
  \\ \ \\
   \displaystyle{
   =\frac{1}{\hbar}+\frac{P_{12}}{u-v}+ \frac{J_1 J_2 P_{12}}{u+v} }.
  \end{array}
\eq
The $R$-matrix (\ref{ratR2}) satisfies the unitarity property
\beq\label{unirat}\begin{array}{c}\displaystyle
    R^{\hbar}_{12}(u,v) R^\hbar_{21}(v,u)=\Big( \frac{1}{\hbar^2}-
  \frac{1}{(u-v)^2}-
  \frac{1}{(u+v)^2}\Big)\rm{Id}\,,
\end{array}\eq
and the skew-symmetry property (\ref{skewsym}).

 \begin{prop}
The $R$-matrix (\ref{ratR2}) satisfies the associative Yang-Baxter equation 
(\ref{AYBE}).
\end{prop}
\begin{proof}
The statement can be proved by direct calculation. Consider the l.h.s. of (\ref{AYBE}) as a monomial of the permutation operators.
For the coefficients of the zero and first power we obtain obvious identities. Non-trivial equations arise for second power only.
For example, let us show that
\beq\label{a21}
\begin{array}{c}
   \displaystyle{
\frac{P_{12}}{u_1-u_2}\cdot \frac{P_{23}}{u_2-u_3}-\frac{P_{13}}{u_1-u_3}\cdot \frac{P_{12}}{u_1-u_2}-\frac{P_{23}}{u_2-u_3}\cdot \frac{P_{13}}{u_1-u_3}=0
}
\end{array}
\eq
and
\beq\label{a22}
\begin{array}{c}
   \displaystyle{
\frac{J_1 J_2 P_{12}}{u_1+u_2}\cdot \frac{J_2 J_3 P_{23}}{u_2+u_3}-\frac{P_{13}}{u_1-u_3}\cdot \frac{J_1 J_2 P_{12}}{u_1+u_2}-\frac{J_2 J_3 P_{23}}{u_2+u_3}\cdot \frac{P_{13}}{u_1-u_3}=0.
}
\end{array}
\eq
To prove (\ref{a21}) we use the following relations in symmetric group 
\beq\label{Prel}
P_{12} P_{23}=P_{13} P_{12}=P_{23} P_{13}
\eq 
and the Fay identity (\ref{Fay}) for $\phi(x,u):=\frac{1}{u}$ and $u=u_1-u_2, v=u_2-u_3$.
Using (\ref{Jrel}) and (\ref{Prel}) we show that the numerators in l.h.s. (\ref{a22}) are equal up to a sign:
$$
J_1 J_2 P_{12} J_2 J_3 P_{23}=P_{12}  J_2 J_1 J_2 J_3 P_{23}=P_{12} J_1 J_3 P_{23} =J_2 J_3 P_{12} P_{23}=J_2 J_3 P_{23} P_{13},
$$
$$
P_{13} J_1 J_2 P_{12}=J_3 J_2 P_{13} P_{12}=-J_2 J_3 P_{13}P_{12}=-J_2 J_3 P_{23}P_{13},
$$
then we use the Fay identity (\ref{Fay}) for $\phi(x,u):=\frac{1}{u}$ and $u=u_2+u_3, v=u_1-u_3$.
\end{proof} 
\begin{prop}
The $R$-matrix (\ref{ratR2}) satisfies the quantum Yang-Baxter equation 
(\ref{QYB}).
\end{prop}
\begin{proof}
The statement follows form the Lemma below and the properties (\ref{unirat}) and (\ref{skewsym}).
\end{proof}
\begin{lemma}\label{Lemma1}
If $R^{\hbar}(u,v)$ is unitary
\beq\label{unilem}
  R^{\hbar}_{12}(u,v) R^\hbar_{21}(v,u)=f^\hbar(u,v)\,\rm{Id}
\eq
with $f^\hbar(u,v)=f^\hbar(v,u)$ and skew-symmetric (\ref{skewsym}) solution of the AYBE (\ref{AYBE}), then it is a solution of the quantum Yang-Baxter equation (\ref{QYB}).
\end{lemma}
The proof for the $R$-matrix depending on the difference of spectral parameters can be found in \cite{LOZ14} (Section 4). Below we almost repeat this proof for our purposes.

\begin{proof}
Consider (\ref{AYBE}) for $x=2\hbar$ and $y=\hbar$:
\beq\label{a221}
  R^{2\hbar}_{12}(u_1,u_2) R^{\hbar}_{23}(u_2,u_3) - R^{\hbar}_{13}(u_1,u_3) R^{\hbar}_{12}(u_1,u_2) - R^{-\hbar}_{23}(u_2,u_3) R^{2\hbar}_{13}(u_1,u_3)=0
\eq
Multiplying (\ref{a221}) by $R_{23}^\hbar(u_2,u_3)$ from the left and using (\ref{unilem}) and (\ref{skewsym})  we obtain: 
\beq\label{a220}
 R_{23}^\hbar(u_2,u_3)R^{\hbar}_{13}(u_1,u_3) R^{\hbar}_{12}(u_1,u_2) =R_{23}^\hbar(u_2,u_3)  R^{2\hbar}_{12}(u_1,u_2) R^{\hbar}_{23}(u_2,u_3) - f^\hbar(u_2,u_3) R^{2\hbar}_{13}(u_1,u_3).
\eq
Now consider (\ref{AYBE}) changing the indices $2\leftrightarrow 3$ and the varibles $u_2\leftrightarrow u_3$ with $x=2\hbar$ and $y=\hbar$:
\beq\label{a222}
  R^{2\hbar}_{13}(u_1,u_3) R^{\hbar}_{32}(u_3,u_2) - R^{\hbar}_{12}(u_1,u_2) R^{\hbar}_{13}(u_1,u_3) - R^{-\hbar}_{32}(u_3,u_2) R^{2\hbar}_{12}(u_1,u_2)=0,
\eq
 multiply (\ref{a222}) by $R_{23}^\hbar(u_2,u_3)$ from the right and use (\ref{unilem}): 
\beq\label{a223}
R^{\hbar}_{12}(u_1,u_2) R^{\hbar}_{13}(u_1,u_3)R_{23}^\hbar(u_2,u_3) =  R^{2\hbar}_{13}(u_1,u_3)f^\hbar(u_3,u_2) - R^{-\hbar}_{32}(u_3,u_2) R^{2\hbar}_{12}(u_1,u_2)R_{23}^\hbar(u_2,u_3) .
\eq
The right hand sides of (\ref{a220}) and (\ref{a223}) are equal due to (\ref{skewsym}), which means (\ref{QYB}).
\end{proof}
\subsection{Semiclassical limit}
In a semiclassical limit $\hbar\to 0$ we have the expansion 
\beq\label{exprat}
R_{12}^\hbar(u,v)=\frac{1}{\hbar}+r_{12}(u,v),
\eq
where
\beq\label{clr}
r_{12}(u,v)=\sum_{i,j} e_{ij}\otimes e_{ji} \frac{(-1)^{p_j}}{u-v}+\sum_{i,j} e_{ij}\otimes e_{-j,-i} \frac{(-1)^{p_j}}{u+v}.
\eq
Here $r_{12}(u,v)\in{{{\rm End}} (\mC^{N|N})^{\otimes 3}}(u,v,w)$ is the classical $R$-matrix for the queer Lie superalgebra \cite{Nazold}. 
Due to (\ref{skewsym}) the solution (\ref{clr}) of (\ref{clYBE}) is antisymmetric: $r_{12}(u,v)=-r_{21}(v,u)$.  
It is known that $\displaystyle \frac{P_{12}}{u-v}$ is the solution of the classical Yang-Baxter equation, (\ref{clr}) can be rewritten 
with the help of (\ref{autom}):
\beq
r_{12}(u,v)=\frac{P_{12}}{u-v}+({\rm Id} \otimes \eta)\frac{P_{12}}{u+v}
\eq
Due to the simple expansion and (\ref{AYBE}) we have the "half" of the classical Yang-Baxter equation for this classical $r$-matrix:
\beq
r_{12}(u,v)r_{23}(v,w)-r_{13}(u,w)r_{12}(u,v)-r_{23}(v,w)r_{13}(u,w)=0.
\eq
\section{Trigonometric R-matrices}\label{trR}

\setcounter{equation}{0}
In \cite{Olsh} (Therorem 4.1)  there  was constructed a solution $S_{12}^\hbar\in {{\rm End}} (\mC^{N|N})\otimes {Q(N)}$ of the quantum Yang-Baxter equation (\ref{QYB}) related to the queer Lie superalgebra, which up to a normalization looks like:
\begin{equation}\label{SS}
 S_{12}^\hbar=\frac{2\pi\imath}{q-q^{-1}}\sum_{i\le j}e_{ij}\otimes s_{ij},
\end{equation}
where 
\beq
  \begin{array}{c}
     \displaystyle{
s_{aa}=1+(q-1)(e_{aa}+e_{-a,-a})} \qquad s_{-a,-a}=1+(q^{-1}-1)(e_{aa}+e_{-a,-a})
\\ \ \\
\displaystyle{
s_{ab}=(q-q^{-1})(e_{ba}+e_{-b,-a}),\, \, a<b,\qquad s_{-b,-a}=-(q-q^{-1})(e_{ba}+e_{-b,-a})\, \, a<b,}
\\ \ \\
\displaystyle{s_{-a,b}=-(q-q^{-1})(e_{b,-a}+e_{-b,a})} \qquad a,b=1,\dots N.
\end{array}
\eq
This solution does not depend on spectral parameter and depend only on $q=\exp{\pi \imath \hbar}$, so the quantum Yang-Baxter equation can be written as:
\beq 
S_{12}^{\hbar} S_{13}^{\hbar} S_{23}^{\hbar}=S_{23}^{\hbar}S_{13}^{\hbar}S_{12}^{\hbar}.
\eq
The solution (\ref{SS}) does not satisfy the associative Yang-Baxter equation. However, we can use the following twist transformation :
\beq\label{Stilde}
\tilde{S}_{12}^\hbar=F_{12}^\hbar S_{12}^{\hbar}\left(F_{21}^{\hbar}\right)^{-1}.
\eq
 where $F_{12}^\hbar\in {\rm Q}(N) \otimes {{\rm Q}(N)}$: 
 
\beq\label{twist}
\displaystyle{
F_{12}^\hbar=\sum_{a,b=1}^N \exp\left(\frac{\pi \imath \hbar \left((a-b)-N{\rm sign}(a-b)\right)}{2N}\right)(e_{aa}+e_{-a,-a})\otimes (e_{bb}+e_{-b,-b})\,.
}
\eq
Then $\tilde{S}_{12}^\hbar$ satisfies both the quantum Yang-Baxter equation (\ref{QYB}) and the associative Yang-Baxter equation (\ref{AYBE}), which looks like:
\beq\label{assforS}
  \tilde{S}^{x}_{12} \tilde{S}^{y}_{23}- \tilde{S}^{y}_{13} \tilde{S}^{x-y}_{12} - \tilde{S}^{y-x}_{23} \tilde{S}^{x}_{13}=0,
\eq
since $ \tilde{S}^{\hbar}_{12} $ does not depend on spectral parameters.

Following  \cite{JNaz} (Section 3) from $S_{12}^{\hbar}$ one can construct a solution $\mR^\hbar_{12}(u,v)$ of the quantum Yang-Baxter equation with spectral parameters in the following way: 
\beq
\mR^\hbar_{12}(u,v)=S_{12}^\hbar+\frac{\pi e^{-\pi \imath(u-v)}}{\sin{\pi(u-v)}} P_{12} + \frac{\pi e^{-\pi \imath(u+v)}}{\sin{\pi(u+v)}} J_1 J_2 P_{12}.
\eq
In standard generators it is given by
 \beq
   \begin{array}{c}\label{uqsuperqueer}
   \displaystyle{
  \mR^\hbar_{12}(u,v)=\frac{\pi}{\sin \pi\hbar}\mathrm{Id}+
   \pi\sum\limits_a\Big((-1)^{p_a}\cot(\pi (u-v))+\coth(\pi\hbar)-\frac{1}{\sin \pi\hbar}\Big)e_{aa}\otimes e_{aa}
 +
  }
  \\ \ \\
   \displaystyle{
 +\frac{\pi}{\sin(\pi (u-v))}\sum\limits_{a< b}
 \Big( (-1)^{p_b} e_{ab}\otimes e_{ba}\,e^{\pi\imath (u-v)}+(-1)^{p_a}e_{ba}\otimes
 e_{ab}\,e^{-\pi\imath (u-v)}\Big)\,+
  }
   \\ \ \\
    \displaystyle{
   +\pi\sum\limits_a\Big((-1)^{p_a}\cot(\pi (u+v))+\coth(\pi\hbar)-\frac{1}{\sin \pi\hbar}\Big)e_{aa}\otimes e_{-a,-a}
 +
  }
  \\ \ \\
   \displaystyle{
 +\frac{\pi}{\sin(\pi (u+v))}\sum\limits_{a< b}
 \Big( (-1)^{p_b} e_{ab}\otimes e_{-b,-a}\,e^{\pi\imath (u+v)}+(-1)^{p_a}e_{ba}\otimes
 e_{-a,-b}\,e^{-\pi\imath (u+v)}\Big)\,.
  }
  \end{array}
  \eq
It can be checked that ${\mR}^{\hbar}_{12}(u,v)$ satisfies the unitarity property
\beq\label{unitrig}\begin{array}{c}\displaystyle
    {\mR}^{\hbar}_{12}(u,v) {\mR}^\hbar_{21}(v,u)=\Big( \frac{\pi^2}{\sin^2(\pi\hbar)}-
  \frac{\pi^2}{\sin^2(\pi(u-v))}-
  \frac{\pi^2}{\sin^2(\pi(u+v))}\Big)\rm{Id}\,,
\end{array}\eq
and the skew-symmetry property (\ref{skewsym}).

Indeed, we can also start from the solution (\ref{Stilde}) and insert spectral parameters:
\beq\label{Rtilde}
\tilde{\mR}^\hbar_{12}(u,v)= \tilde{S}_{12}^\hbar+\frac{\pi e^{-\pi \imath(u-v)}}{\sin{\pi(u-v)}} P_{12} + \frac{\pi e^{-\pi \imath(u+v)}}{\sin{\pi(u+v)}} J_1 J_2 P_{12}.
\eq

Since $F_{12}(\hbar)\in {\rm Q}(N) \otimes {{\rm Q}(N)}$, it commutes with $J_1$ and $J_2$ by definition of $Q(N)$, so we have the following:
\beq
F_{12}^\hbar P_{12}\left(F_{21}^{\hbar}\right)^{-1}=F_{12}^\hbar\left(F_{12}^{\hbar}\right)^{-1}P_{12}=P_{12},
\eq
\beq
F_{12}^\hbar J_1 J_2 P_{12}\left(F_{21}^{\hbar}\right)^{-1}=J_1 J_2P_{12},
\eq
and the $R$-matrices $\mR^\hbar_{12}(u,v)$ and $\tilde{\mR}^\hbar_{12}(u,v)$ are connected by the twist (\ref{twist}):
\beq\label{twistR}
\tilde{\mR}^\hbar_{12}(u,v)=F_{12}^\hbar \mR^\hbar_{12}(u,v)\left(F_{21}^{\hbar}\right)^{-1}.
\eq
From (\ref{twistR}) follows that  $\tilde{\mR}^\hbar_{12}(u,v)$ satisfies the unitarity property (\ref{unitrig}) and the skew-symmetry property (\ref{skewsym}).
The $R$-matrix ${\rm R}^\hbar(u,v)$ from (\ref{superRqueer}) and $\tilde{\mR}^\hbar(u,v)$ from (\ref{Rtilde}) are connected by the gauge transformation:
\beq\label{gauge}
{\rm R}_{12}^\hbar(u,v)=G_1(u)G_2(v)\tilde{\mR}^\hbar_{12}(u,v)G_1^{-1}(u)G_2^{-1}(v),
\eq
where 
\beq\label{G}
\displaystyle G(u)=\sum_j \exp \left(\frac{ \pi \imath u(j-1)}{N}\right)e_{jj}.
\eq
\begin{prop}
The $R$-matrix (\ref{superRqueer}) and (\ref{Rtilde}) satisfy the associative Yang-Baxter equation 
(\ref{AYBE}).
\end{prop}
\begin{proof}
Since the associative Yang-Baxter equation 
(\ref{AYBE}) is invariant under the gauge transformation (\ref{gauge}) it is enough to prove the statement for ${\rm R}^\hbar(u,v)$ from (\ref{superRqueer}).
Let us rewrite (\ref{superRqueer}) as the sum
\beq
{\rm R}_{12}^\hbar(u,v)=D_{12}^\hbar +Q_{12}(u,v),
\eq
where 
 \begin{equation}\label{D}
   \begin{array}{c}
   \displaystyle{
  D_{12}^\hbar
  =\pi\sum\limits_{i}\cot(\pi\hbar) \left(e_{ii}\otimes e_{ii}+e_{ii}\otimes e_{-i,-i}\right)+
   }
  \\ \ \\
     \displaystyle{
  +\pi\sum\limits_{|i|\neq |j|} e_{ii}\otimes e_{jj}\frac{ \exp\Big(\frac{\pi\imath\hbar}{N}\Big(\left(|i|-|j|\right)-N{\rm sign}\left(|i|-|j|\right)\Big)\Big)}{\sin(\pi\hbar)}
  }
    \end{array}
  \end{equation}
 depends only on $\hbar$ and
 \begin{equation}\label{Q}
   \begin{array}{c}
   \displaystyle{
  Q_{12}(u,v)
  =\pi\sum\limits_{i}(-1)^{p_i}\Big(\cot(\pi (u-v)) e_{ii}\otimes e_{ii}+
\cot(\pi (u+v) e_{ii}\otimes e_{-i,-i}\Big)+
   }

  \\ \ \\
   \displaystyle{
 +\pi\sum\limits_{i\neq j} (-1)^{p_j} e_{ij}\otimes e_{ji}\,\frac{ \exp\Big(\frac{\pi\imath (u-v)}{N}\Big((i-j)-N{\rm sign}(i-j)\Big)\Big)}{\sin(\pi (u-v))}\,+
  }
   \\ \ \\
   \displaystyle{
 +\pi\sum\limits_{i\neq j} (-1)^{p_j} e_{i,j}\otimes e_{-j,-i}\,\frac{ \exp\Big(\frac{\pi\imath (u+v)}{N}\Big((i-j)-N{\rm sign}(i-j)\Big)\Big)}{\sin(\pi (u+v))}\,.
  }
  \end{array}
  \end{equation}
  depens on $u,v$. Then the AYBE (\ref{AYBE}) can be rewritten as follows:
  \beq\label{AYBEDQ}
\begin{array}{c}
    (D_{12}^x+Q_{12}(u_1,u_2))( D_{23}^y+Q_{23}(u_2,u_3)) -(D_{13}^y+ Q_{13}(u_1,u_3))( D_{12}^{x-y}+Q_{12}(u_1,u_2)) 
    \\ \ \\
    - (D_{23}^{y-x}+Q_{23}(u_2,u_3))(D_{13}^{x}+Q_{13}(u_1,u_3))=0\,.
\end{array}
\eq
Then the statement can be proved by direct calculation checking the following five identities:
\beq
D_{12}^xD_{23}^y-D_{13}^yD_{12}^{x-y}-D_{23}^{y-x}D_{13}^{x}=-\pi^2\sum_{|i|=|j|=|k|} e_{ii}\otimes e_{jj}\otimes e_{kk}\,,
\eq
\beq
\begin{array}{c}
Q_{12}(u_1,u_2)Q_{23}(u_2,u_3)-Q_{13}(u_1,u_3)Q_{12}(u_1,u_2)-Q_{23}(u_2,u_3)Q_{13}(u_1,u_3)=
\\ \ \\
\displaystyle{
=\pi^2\sum_{|i|=|j|=|k|} e_{ii}\otimes e_{jj}\otimes e_{kk}}\,,
\end{array}
\eq
\beq
D_{12}^x Q_{23}(u_2,u_3)=Q_{23}(u_2,u_3)D_{13}^x \,,
\eq
\beq
D_{13}^y Q_{12}(u_2,u_3)= Q_{12}(u_2,u_3) D_{23}^y\,,
\eq
\beq
D_{23}^{y-x}Q_{13}(u_1,u_3)=-Q_{13}(u_1,u_3)D_{12}^{x-y}.
\eq
\end{proof}
\begin{prop}
The $R$-matrix (\ref{superRqueer}) satisfies the quantum Yang-Baxter equation 
(\ref{QYB}).
\end{prop}

\begin{proof}
The statement follows from the Lemma \ref{Lemma1} and the properties (\ref{unitrig}) and (\ref{skewsym}).
\end{proof}
\\

The $R$-matrix (\ref{uqsuperqueer}) does not satisfy the AYBE (\ref{AYBE}), but it satisfies the associative Yang-Baxter equation with the additional term which does not depend on spectral parameters $u_1, u_2, u_3$:
\begin{prop}
The $R$-matrix ${\mR}^{\hbar}(u,v)$  (\ref{uqsuperqueer}) satisfies the following equation:
    \beq\label{relAY2}
\begin{array}{c}
\displaystyle{
    \mR^{x}_{12}(u_1-u_2) \mR^{y}_{23}(u_2-u_3) - \mR^{y}_{13}(u_1-u_3) \mR^{x-y}_{12}(u_1-u_2) - \mR^{y-x}_{23}(u_2-u_3) \mR^{x}_{13}(u_1-u_3)=
    }
    \\ \ \\
    \displaystyle{
    =\frac{\pi^2}{2\cos(\frac{\pi x}{2})\cos(\frac{\pi y}{2})\cos(\frac{\pi (x-y)}{2})}\sum_{|i|\neq |j|\neq |k|\neq |i|} e_{ii}\otimes e_{jj}\otimes e_{kk}\,.
    }
\end{array}\eq
\end{prop}
\begin{proof}
The proof is similar to the case of $\tilde{\mR}^{\hbar}_{12}(u,v)$. It is enough to show that
$$
  {S}^{x}_{12} {S}^{y}_{23}- {S}^{y}_{13} {S}^{x-y}_{12} - {S}^{y-x}_{23} {S}^{x}_{13}=\frac{\pi^2}{2\cos(\frac{\pi x}{2})\cos(\frac{\pi y}{2})\cos(\frac{\pi (x-y)}{2})}\sum_{|i|\neq |j|\neq |k|\neq |i|} e_{ii}\otimes e_{jj}\otimes e_{kk}.
  $$
\end{proof}\\

Note that the right hand side of (\ref{relAY2}) is the diagonal element which belongs to ${\rm Q}(N) \otimes {{\rm Q}(N)}\otimes{{\rm Q}(N)} $.
\subsection{Classical limit}
 In a semiclassical limit $\hbar\to 0$ we have the expansion 
\beq\label{exptrig}
{\rm R}_{12}^\hbar(u,v)=\frac{\pi }{\hbar}+{\rm r}_{12}(u,v)+\hbar\, {\rm m}_{12}+O(\hbar^2),,
\eq
where
\beq\label{clrtrig}
{\rm r}_{12}(u,v)=Q_{12}(u,v)+\pi\sum\limits_{|i|\neq |j|}  \frac{ \pi \imath}{N}\Big(\left(|i|-|j|\right)-N{\rm sign}\left(|i|-|j|\right)\Big) e_{ii}\otimes e_{jj}
\eq 
\beq\label{m12}
 \begin{array}{c}
   \displaystyle{
  {\rm m}_{12}
  =-\frac{\pi^2}{3}\sum\limits_{i} \left(e_{ii}\otimes e_{ii}+e_{ii}\otimes e_{-i,-i}\right)+
   }
  \\ \ \\
     \displaystyle{
 +\frac{\pi^2}{6}\sum\limits_{|i|\neq |j|} e_{ii}\otimes e_{jj} -\frac{\pi^2}{2 N^2}\sum\limits_{|i|\neq |j|} \Big(\left(|i|-|j|\right)-N{\rm sign}\left(|i|-|j|\right)\Big)^2 e_{ii}\otimes e_{jj}.
  }
    \end{array}
    \eq
  The classical $r$-matrix is antisymmetric 
$
{\rm r}_{12}(u,v)=-{\rm r}_{21}(v,u).
$
From (\ref{AYBE}) we have :
\beq\label{hclYB}
{\rm r}_{12}(u,v){\rm r}_{23}(v,w)-{\rm r}_{13}(u,w){\rm r}_{12}(u,v)-{\rm r}_{23}(v,w){\rm r}_{13}(u,w)=-{\rm m}_{12}-{\rm m}_{23}-{\rm m}_{13}.
\eq
Changing the indices $2\leftrightarrow 3$ in (\ref{hclYB}) we have :
\beq\label{hclYB2}
{\rm r}_{13}(u,w){\rm r}_{32}(w,v)-{\rm r}_{12}(u,v){\rm r}_{13}(u,w)-{\rm r}_{32}(w,v){\rm r}_{12}(u,v)=-{\rm m}_{13}-{\rm m}_{32}-{\rm m}_{12}.
\eq
Substracting (\ref{hclYB2}) from (\ref{hclYB}) we obtain the classical Yang-Baxter equation (\ref{clYBE})

\subsection*{Acknowledgments}

I am very greatful to Maxim Nazarov for helpful discussions and valuable remarks. I am also grateful to Andrei Zotov and Andrii Liashyk for their kind interest in this work.\\
This work was supported by the Russian Science Foundation under grant no. 23-11-00150.

\newpage 
\begin{small}

\end{small}

\end{document}